\documentclass{article}
\usepackage[utf8]{inputenc}
\usepackage[top=1in,left=1in,right=1in,bottom=1in]{geometry}

\usepackage{amsthm}
\usepackage{amsmath}
\usepackage{amssymb}
\usepackage{relsize}
\usepackage[noend,ruled,linesnumbered]{algorithm2e}
\SetKwProg{Fn}{Function}{}{}
\usepackage[noend]{algpseudocode}
\usepackage{multicol}
\usepackage{relsize}
\usepackage[most]{tcolorbox}
\usepackage{color}
\usepackage{thmtools}
\usepackage{thm-restate}
\usepackage[colorlinks=true,linkcolor=blue,citecolor=blue,allcolors=blue]{hyperref}
\hypersetup{
    colorlinks=true,
    linkcolor=blue,
    filecolor=blue,
    urlcolor=blue,
    allcolors=blue
}

\usepackage{mathtools}
\usepackage{enumerate}
\usepackage{pgfplots}
\usepackage{subcaption}
\usepackage{graphicx}
\usepackage{booktabs}
\usepackage{xspace}

\newcommand{\congest}{\textsc{Congest}\xspace}

\newcommand{\del}{\textsc{Delete}\xspace}

\newcommand{\listc}{\mathsf{List}\xspace}
\newcommand{\listbatched}{\mathsf{BatchedList}\xspace}

\newcommand{\id}{\textsc{ID}\xspace}

\newcommand{\detect}{\mathsf{Detect}\xspace}

\newcommand{\local}{\textsc{Local}\xspace}

\newcommand{\new}{\textsc{New}\xspace}

\definecolor{mygreen}{RGB}{20,140,80}
\definecolor{linkcolor}{RGB}{0,0,230}
\definecolor{mylightgray}{RGB}{230,230,230}
\definecolor{verylightgray}{RGB}{245,245,245}

\hypersetup{
     colorlinks=true,
     citecolor= mygreen,
     linkcolor= mygreen,
     urlcolor= mygreen
}

\newcommand{\etal}[0]{\textit{et al.}}

\newcounter{myalgctr}

\newtcolorbox{OuterBox}[1][]{%
    breakable,
    enhanced,
    frame hidden,
    interior hidden,
    left=-5pt,
    right=-5pt,
    top=-5pt,
    float=p,
    boxsep=0pt,
    arc=0pt
#1}%

\newtcolorbox{InnerBox}[1][]{%
    enforce breakable,
    enhanced,
    colback=gray,
    colframe=white,
#1}%

\newenvironment{tbox}{
\vspace{0.2cm}
\begin{tcolorbox}[width=\columnwidth,
                  enhanced,
                  boxsep=2pt,
                  left=1pt,
                  right=1pt,
                  top=4pt,
                  boxrule=1pt,
                  arc=0pt,
                  colback=white,
                  colframe=black,
	              breakable
                  ]%
}{
\end{tcolorbox}
}

\newcommand{\tboxhrule}[0]{\vspace{0.1cm} {\color{black} \hrule} \vspace{0.2cm}}

\newenvironment{titledtbox}[1]{\begin{tbox}#1 \tboxhrule}{\end{tbox}}

\algrenewcommand\algorithmicindent{1em}%

\newcommand{\tO}{\widetilde{O}}

\usepackage[capitalize, nameinlink]{cleveref}

\theoremstyle{plain}
\newtheorem{theorem}{Theorem}[section]

\newtheorem{corollary}[theorem]{Corollary}

\newtheorem{observation}[theorem]{Observation}

\crefname{theorem}{Theorem}{Theorems}
\Crefname{lemma}{Lemma}{Lemmas}
\Crefname{claim}{Claim}{Claims}
\Crefname{observation}{Observation}{Observations}
\Crefname{algorithm}{Algorithm}{Algorithms}
\Crefname{myalgctr}{Algorithm}{Algorithms}
\Crefname{challenge}{Challenge}{Challenges}
\Crefname{thm}{Theorem}{Theorems}

\begin{document}
\title{A Note on Improved Results for One Round Distributed Clique Listing}

\author{Quanquan C. Liu\thanks{Northwestern University, quanquan@northwestern.edu}}
\date{}

\maketitle

\begin{abstract}
In this note, we investigate listing cliques
of arbitrary sizes in bandwidth-limited, dynamic networks. The problem of
detecting and listing triangles and cliques was
originally studied in great detail by Bonne and Censor-Hillel (ICALP 2019).
We extend this study to dynamic graphs where more than one update may occur
as well as resolve an open question posed by Bonne and Censor-Hillel (2019).
Our algorithms and results are based on some simple observations
about listing triangles
under various settings and we show that we can list larger cliques using
such facts.
Specifically, we show that our techniques can be used to solve an open problem posed in the
original paper: we show that detecting and listing cliques (of any size)
can be done using $O(1)$-bandwidth after one round of communication
under node insertions and node/edge deletions. We conclude with an extension of our
techniques to obtain a small bandwidth $1$-round algorithm for listing
cliques when more than one node insertion/deletion and/or edge deletion
update occurs at any time.
\end{abstract}

\section{Introduction}
Detecting and listing subgraphs under limited bandwidth conditions is a fundamental
problem in distributed computing.
Although the static version of this problem has been studied by many
researchers in the past~\cite{Abboud17,CCG21,CGL20,
CPSZ21,ChangPZ19,CS19,CS20,
DruckerKO13, EFFKO19,FischerGKO18,FraigniaudMORT17,GonenO17,
HPZZ20,IzumiG17,KorhonenR17,Pandurangan18}
in both the upper bound and lower bound
settings, the dynamic version of such problems
often require different techniques. Triangle and clique listing are
problems where every occurrence of a triangle or clique in the graph is listed by
at least one node in the triangle or clique.
More specifically, the question
we seek to answer is: given a change in the topology
of the graph under one or more updates,
can we accurately list all cliques in the updated graph?
For certain settings, this question has an easy answer.
For example, in the \local model where communication is synchronous and error-free
and messages can have unrestricted size,
it is trivial for any node to list all $k$-cliques
adjacent to it after any set of edge insertion/deletion
and/or node insertion/deletion. 
In the \local model, every node would broadcast the
entirety of its adjacency list to all its neighbors each round.
Thus, each node can reconstruct the edges between its neighbors from these 
messages and it is easy to solve the clique listing problem using the reconstructed
neighborhood.

In the traditional \congest model, messages are passed between neighboring 
nodes in synchronous rounds where each message has size $O(\log n)$.
Detecting and listing triangles and cliques
in the \congest model turn out to be much harder problems.
This note focuses on triangle and clique \emph{listing} for which a number of previous
works provided key results. The summary of these results can be found in~\cref{table:congest-results}.
In the table, $\Delta$ is the maximum degree in the input graph. All of these results focus 
on the \emph{static} setting where the topology of the graph does not change.

\begin{table}[htb!]
    \centering
    \footnotesize
    \begin{tabular}{| c | c |}
    \toprule
         Problem & Rounds \\
         \hline
         Triangle Listing & \shortstack{$\tilde{O}(n^{3/4})$~\cite{IzumiG17}\\ $\tilde{O}(n^{1/2})$~\cite{ChangPZ19} \\
         $\tilde{O}(n^{1/3})$~\cite{CS19}}\\
         \hline
         Triangle Listing Lower Bound & $\tilde{\Omega}(n^{1/3})$~\cite{Pandurangan18,IzumiG17}\\
         \hline
         Deterministic Triangle Listing & \shortstack{$n^{2/3 + o(1)}$~\cite{CS20}\\
         $O(\Delta/\log n + \log \log \Delta)$~\cite{HPZZ20}}\\
         \hline
         $4$-Clique Listing & $n^{5/6 + o(1)}$~\cite{EFFKO19}\\
         \hline
         $4$-Clique Listing Lower Bound & $\tilde{\Omega}(n^{1/2})$~\cite{CzumajK20}\\
         \hline
         $5$-Clique Listing & \shortstack{$n^{73/75 + o(1)}$~\cite{EFFKO19}\\
         $n^{3/4 + o(1)}$~\cite{CGL20}}\\
         \hline
         $k$-Clique Listing & \shortstack{$\tilde{O}(n^{k/(k+2)})$ ($k \geq 4, k \neq 5$)~\cite{CGL20}\\
         $\tO(n^{1-2/k})$~\cite{CCG21}\\
         $n^{1-2/k + o(1)}$~\cite{CLV22}}\\
         \hline
         $k$-Clique Listing Lower Bound & \shortstack{$\tilde{\Omega}(n^{1-2/k})$~\cite{FischerGKO18}\\
         $\tilde{\Omega}(n^{1/2}/k)$ ($k \leq n^{1/2}$)~\cite{CzumajK20}\\
         $\tilde{\Omega}(n/k)$ ($k > n^{1/2}$)~\cite{CzumajK20}}\\
    \bottomrule
    \end{tabular}
    \caption{Results for subgraph listing problems in the \congest model.}
    \label{table:congest-results}
\end{table}

Additional works also provide lower bounds
under very small bandwidth settings~\cite{Abboud17,
IzumiG17,FischerGKO18,Pandurangan18}. A detailed description of all the aforementioned results
can be found in the recent
comprehensive survey of Censor-Hillel~\cite{censorhillel21}.

In this paper, we focus on the \emph{dynamic} setting
for subgraph listing problems
for which we are able to obtain $1$-round
algorithms that require $O(1)$ or $O(\log n)$ bandwidth.
In the dynamic setting, edge and node updates are applied to an initial (potentially empty)
graph. The updates change the topology of the graph.
All nodes in the initial graph know the graph's complete topology.
After the application of each round of updates, the nodes collectively
must report an accurate list of
the correct set of triangles or cliques.
We build off the elegant results of~\cite{BC19} who define
and investigate in great depth the question of detecting
and listing triangles and cliques in dynamic networks.
The model they present in their paper can capture
real-world behavior such as nodes joining or leaving the network
or communication links which appear or disappear between pairs of
nodes at different points in time.
In this paper, we make several observations about triangle
listing which may be used to re-prove results provided in their
paper as well as resolve an open question stated in their work.
Our paper is structured as follows. First, we formally define the model
in~\cref{sec:prelims}. Then, we summarize our results in~\cref{sec:contributions}.
\cref{sec:clique} gives our main result for listing $k$-cliques under 
node insertions/deletions and edge deletions. \cref{sec:wedges}
describes our result for listing wedges under node deletions and 
edge insertions/deletions. \cref{sec:batched-triangle} describes 
our result on triangle listing under batched updates, and~\cref{sec:batched-clique}
gives our result on clique listing under batched updates.
\section{Preliminaries}\label{sec:prelims}
The model we use in our paper is the same as the model used in~\cite{BC19}.
For completeness, we restate their model as well as the problem
definitions in this section.

The network we consider in this dynamic setting can be modeled as a sequence of graphs: $(G_0, \dots, G_r)$. 
The initial graph $G_0$ represents the starting state of the network. All nodes in the
initial graph $G_0$ know the graph's complete topology. Each subsequent graph in the sequence
is either identical to the preceding graph or differs from it by a single topology
change: either an edge insertion, edge deletion, node insertion (along with its
adjacent edges), or node deletion (along with its adjacent edges). Later, we
also consider graphs where each node may be adjacent to $O(1)$ multiple topology
changes. We denote the neighbors of node $v$ in the $i$-th graph
by $N_i(v)$; we omit the subscript $i$ when the current graph is 
obvious from context.

We assume the network is synchronized. In each round, each node can send to
each one of its neighbors a message containing $B$ bits where $B$ is defined as
the \emph{bandwidth} of the network. The messages sent to different neighbors 
can be different. In this paper, we focus on problems that
can be solved with bandwidth $B = O(1)$ or $B = O(\log n)$.

We assume that each node has a unique ID and each node in the network knows
the IDs of all its neighbors. In other words, we assume that each node
has an adjacency list containing the IDs of all its neighbors and knows
when a new node (previously not in its adjacency list) becomes connected
to it. Furthermore, we assume that the size of any ID is $O(\log{n})$ in
bits. Since in both~\cite{BC19} and our work, the number of vertices in the
graph can change dynamically, we define $n$ to be the number of vertices in the
graph in the current round.

Each round proceeds in three synchronous parts as follows. 

\begin{enumerate}
    \item At the start of each round, the topological change occurs. 
    Nodes are inserted/deleted from adjacency lists and new communication 
    links are established or destroyed between pairs of nodes. A node can
    determine whether an adjacent update occurred on an adjacent
    edge by comparing its current list of neighbors with its list of neighbors
    from the previous round. However, this also
    means that a node $u$ cannot distinguish between
    the insertion of an edge $(u, v)$ and the insertion of node $v$.
    \item Then, nodes exchange messages using (potentially new) communication links. 
    Nodes previously able to communicate
    may not continue to be able to communicate once the links are destroyed due to new deletions.
    \item Finally, nodes receive messages and list the triangles or cliques for this round.
\end{enumerate}

In this paper, we only deal with $1$-round algorithms or algorithms where
the output of each node in the graph following $1$ round of communication
after the last topological change is correct. We define the \emph{$1$-round
bandwidth complexity} of an algorithm to be the minimum bandwidth $B$
for which such a $1$-round algorithm exists. Our paper only deals with
\emph{deterministic} $1$-round algorithms.

We state a subgraph is \emph{created} if it is a subgraph that is created by
a new edge in the current round.
We state a subgraph is \emph{destroyed} if it is listed by at least one node in the 
previous round but no longer exists in the current round. When a node is deleted,
it can no longer send any messages or list any subgraphs in the round where
it is deleted (because it no longer exists). 

We solve the triangle and $k$-clique \emph{listing} problem in $O(1)$ or $O(\log n)$
bandwidth in this paper. We denote a clique with $s$ vertices by $K_s$.
Specifically, the problem $\listc(H)$
is defined to be the problem where given a generic unlabeled graph $H$, each labeled subgraph of the current
graph $G_i \in (G_0, G_1, \dots, G_r)$ that is isomorphic to $H$ must be listed by at least one node. 
A node \emph{lists} a subgraph if it lists the IDs of all nodes in the subgraph.
Furthermore, every listed subgraph must be isomorphic to $H$.
Naturally, all of our listing algorithms also apply to the \emph{detection}
setting where at least one node in the network detects the appearance of one or more $H$
(the problem denoted by $\detect(H)$ in~\cite{BC19}).
As shown in Observation 1 of~\cite{BC19},
$B_{\detect(H)} \leq B_{\listc(H)}$, and we only discuss
$\listc(H)$ in the following sections whose results
also transfer to $\detect(H)$.

\subsection{Batched Updates Model} 
We define a new model where more than one update can occur in each round. 
Specifically, for every node $v \in V_i$ in the current graph $G_i = (V_i, E_i)$, 
at most $O(1)$ updates are incident to $v$. This means that any node in 
$G_{i-1} = (V_{i-1}, E_{i-1})$ is incident to $O(1)$ new edge insertions and deletions in $G_i$. A node 
$u \in V_{i} \setminus V_{i-1}$ is a newly inserted node and \emph{all} edges 
incident to $u$ are newly inserted edges. Hence, $u$ is incident to $O(1)$ edges
when it is inserted. We denote the subgraph $H$ listing problem in this model
as $\listbatched(H)$. 

Notably, we do not need the assumption that each newly inserted node $v$ can tell which of its neighbors are also newly inserted nodes. 
This is because all edges incident to newly inserted nodes are newly inserted edges. Thus, a newly inserted node $v$ 
can send its entire adjacency list to all its neighbors. 
An important part of our model is 
that all nodes, including the newly inserted nodes, are adjacent to $O(1)$ updates so newly inserted nodes can send its adjacency list in 
$O(\log n)$ bandwidth to all its neighbors.

\section{Our Contributions}\label{sec:contributions}
All of the algorithms in our paper are deterministic $1$-round 
algorithms that use $O(1)$ or $O(\log n)$ bandwidth.
Specifically, we show the following main results in this paper.

We first answer Open Question 4 posed in~\cite{BC19}. The previous algorithm
given in~\cite{BC19} under this setting required $O(\log n)$ bandwidth.

\begin{restatable}{thm}{insertsthm}\label{thm:list-clique}
    The deterministic $1$-round bandwidth complexity of listing cliques,
    $\listc(K_s)$, under node insertions and node/edge deletions is $O(1)$.
\end{restatable}

To prove this result, we also show a number of simple but new observations about triangle
listing in~\cref{sec:clique} which may prove to be useful
for more future research in this area.

In addition to cliques, we give an algorithm for wedges, which has not been considered
in previous literature.
An \emph{induced wedge} among a set of three nodes $\{u, v, w\}$ 
in the input graph 
is a path of length $2$ that uses all three nodes, and the induced subgraph of the three
nodes is not a cycle. 
We denote a wedge by $\Gamma$.

\begin{restatable}{thm}{wedgethm}\label{lem:wedge-listing}
    The deterministic $1$-round bandwidth complexity of listing wedges, 
    $\listc(\Gamma)$, under edge insertions and node/edge deletions,
    is $O(\log n)$.
\end{restatable}

Censor-Hillel \etal~\cite{CKS21} studied the setting of \emph{highly} dynamic
networks in which the number of topology changes per round is \emph{unlimited}.
In this setting, they showed $O(1)$-amortized round complexity algorithms for
$k$-clique listing and four and five cycle listing. In our $1$-round \congest setting,
we show the following two theorems when multiple updates occur in the same round.
Our results are the first to consider more than one update in the $1$-round, low bandwidth
setting.

\begin{restatable}{thm}{triangleslist}\label{lem:batched}
    The deterministic $1$-round bandwidth complexity of listing triangles, $\listbatched(K_3)$,
    when each node is incident to $O(1)$ updates, 
    is $O(\log n)$ under node/edge insertions and node/edge deletions.
\end{restatable}

\begin{restatable}{thm}{cliqueslist}\label{lem:batched-clique}
    The deterministic $1$-round bandwidth complexity of listing cliques of size $s$, $\listbatched(K_s)$,
    when each node is incident to $O(1)$ updates, 
    is $O(\log n)$ under node insertions and node/edge deletions.
\end{restatable}

\section{Clique Detection and Listing}\label{sec:clique}

In this section, we show our main result by providing a simple upper bound for the $1$-round
bandwidth complexity of $\listc(K_s)$ under node insertions
and node/edge deletions. Specifically, we show
that $\listc(K_s)$ has $1$-round $O(1)$-bandwidth complexity.
This directly resolves Open Question $4$
in~\cite{BC19}. We describe our algorithm below and give its pseudocode in~\cref{alg:list-cliques}.
To begin as a simple warm-up
illustration of our technique, we re-prove
Theorem 3.3.3 of~\cite{BC19} using only the node insertion portion of the algorithm.

\SetKwFunction{FnListCliques}{ListCliques}
\begin{algorithm}[!t]
    \caption{\label{alg:list-cliques} Listing Cliques}
    \textbf{Input:} Update $U$ which can be a node insertion, node deletion, or edge deletion.\\ 
    \textbf{Output:} Each $K_s$ is listed by at least one of its nodes and no set of 
    $s$ nodes which is not a $K_s$ is wrongly listed.\\
    \Fn{\FnListCliques{$U$}}{
        \If{$U$ adds $u$ to $N(v)$}{\label{clique:add-adj}
            $v$ sends $\new$ to all $w \in N(v)$.\label{clique:send-new}\\
        } 
        \If{$w$ receives $\new$ from $v$ and $u$ is added to $N(w)$}{\label{clique:receive-new}
            $w$ starts listing new triangle $\{u, v, w\}$.\label{clique:start-listing}\\
        }
        \If{$U$ deletes $u$ from $N(v)$}{\label{clique:destroy-adj}
            $v$ sends $\del$ to all $w \in N(v)$.\label{clique:send-del}\\
            \For{$w \in N(v)$}{
                \If{$v$ lists triangle $\{u, v, w\}$}{\label{clique:list-triangle}
                    $v$ stops listing triangle $\{u, v, w\}$.\label{clique:stop-listing}
                }
            }
        }
        \If{$v$ receives $\del$ from $u$ and $w$}{\label{clique:receive-2-del}
            \If{$v$ lists triangle $\{u, v, w\}$}{\label{clique:list-new-2}
                $v$ stops listing triangle $\{u, v, w\}$.\label{clique:stop-list-2}
            }
        }
        \For{all $v \in V$}{\label{clique:iterate-node}
            \For{every $S \subseteq N(v)$ of $s-1$ neighbors}{\label{clique:every-subset}
                \If{$v$ lists $\{v\} \cup A$ as a triangle for every subset of $2$ nodes $A = \{a, b\}\subseteq S$ }{\label{clique:all-subsets-triangle}
                    $v$ lists $S \cup \{v\}$ as a $K_s$.\label{clique:list-new-clique}
                }
            }
            $v$ stops listing every clique containing a destroyed triangle.\\
        }
    }
\end{algorithm}

In~\cref{alg:list-cliques}, for every update which adds a node $u$ to an adjacency list $N(v)$ (\cref{clique:add-adj}), 
$v$ sends an $O(1)$ sized message to every neighbor $w \in N(v)$ (\cref{clique:send-new}). If any neighbor $w$ also 
gets $u$ added to its adjacency list $N(w)$ (\cref{clique:receive-new}), then $w$ lists triangle $\{u, v, w\}$.
If instead $u$ is deleted from $N(v)$ (\cref{clique:destroy-adj}), then $v$ first sends $\del$ to all its neighbors
$w \in N(v)$ (\cref{clique:send-del}). Furthermore, for every neighbor $w \in N(v)$, 
if $v$ lists triangle $\{u, v, w\}$ (\cref{clique:list-triangle}),
then $v$ stops listing $\{u, v, w\}$ (\cref{clique:stop-listing}). If a node $v$ receives a $\del$ message from two 
of its neighbors (\cref{clique:receive-2-del}) and if it lists $\{u, v, w\}$ as a triangle (\cref{clique:list-new-2}),
then it stops listing $\{u, v, w\}$ (\cref{clique:stop-list-2}). Finally, after receiving all the messages for this round,
every vertex $v \in V$ determines the subsets $S \subseteq N(v)$ of $s-1$ of its neighbors (\cref{clique:every-subset}) where 
every subset of $2$ nodes of $\{a, b\} \subseteq S$ including itself, $\{a, b\} \cup \{v\}$, 
is a triangle (\cref{clique:all-subsets-triangle}); each such $S \cup \{v\}$ is a
$K_s$ and $v$ lists it as a $K_s$. Note that~\cref{alg:list-cliques} treats a node
that is deleted and reinserted like a new node when it is reinserted (i.e.\ 
nodes do not distinguish between a completely new neighbor from one that is deleted and reinserted
at a later time). 

We now prove a set of useful observations which we use to prove our main theorem regarding $\listc(K_s)$. The first
observation below states that any node inserted in round $r$ will list every triangle containing it that 
is formed after its insertion. This is the crux of our analysis since this combined with our other
observations states that every clique \emph{is listed by its earliest inserted node.}

\begin{observation}\label{lem:node-insertion-triangle}
    Under node insertions and node/edge deletions,
    suppose $v$ is inserted in round $r$ and $u$ is inserted in round
    $r' > r$. If $\{u, v, w\}$ is a triangle in round $r'$, then in 
    round $r'$, $u$ lists $\{u, v, w\}$ using $O(1)$ bandwidth. 
\end{observation}

\begin{proof}
    We prove this via induction. We assume our observation holds
    for the $j$-th node insertion and prove it for the
    $(j+1)$-st insertion. The observation trivially holds for the
    initial graph since all nodes know the topology of the
    initial graph and thus lists all triangles they are
    part of. Let $v$ be a node inserted in
    round $r \leq j$ and $u$ be a node
    inserted in round $j + 1$. If $u$ forms a triangle with
    $v$ and another node $w$, when node $u$ is inserted, it
    must contain an edge to both $v$ and $w$ and edge $\{v, w\}$
    must already exist in the graph. Edge $\{v, w\}$ must already
    exist in the graph since edges are only inserted as a
    result of node insertions; hence, if an edge between $v$ and
    $w$ exists it must have been inserted when either $v$ or $w$
    was inserted (or it existed in the initial graph),
    whichever one was inserted later. When a node
    receives a new neighbor, it sends $\new$ to all its neighbors 
    indicating it received a new neighbor (\cref{clique:send-new} of~\cref{alg:list-cliques}). 
    Since both $v$ and $w$ are adjacent and send
    $\new$ to each other, $v$ receives $\new$
    from $w$ and knows that $u$ must be its new neighbor. Thus,
    $v$ lists the triangle containing $w$ and $u$. Because this holds
    for all vertices inserted in round $r \leq j$, we proved that every 
    node inserted in round $r \leq j$ correctly lists every triangle 
    created in round $j + 1$ that contains it.

    We now show that $v$ does not incorrectly list a set of three nodes $\{u, v, w\}$
    as a triangle when the set is not a triangle. This means that we show on an edge
    or node deletion, $v$ stops listing triangles that are destroyed.
    We first show that any node which
    lists a triangle will also successfully list its deletion
    after an edge deletion. First, if an edge deletion removes a neighbor of 
    $v$ that is part of a triangle that $v$ lists, then $v$ stops listing this triangle (\cref{clique:stop-listing}).
    Node $v$ also sends $\del$ to all its neighbors (\cref{clique:send-del}).
    Suppose the edge deletion happens between $u$ and $v$ and $w$ lists triangle $\{u, v, w\}$.
    Then, $w$ receives $\del$ from $u$ and $v$ in round $j + 1$ and knows that edge $\{u, v\}$
    is deleted; $w$ then stops listing $\{u, v, w\}$ (\cref{clique:stop-list-2,clique:list-new-2,clique:receive-2-del}).
    Hence, after an edge deletion in round $j + 1$, all nodes which lists a triangle destroyed
    by the deletion stops listing the triangle.
    What remains to be shown is that a node which lists
    a triangle can successfully list its deletion after a node
    deletion. When a triangle is destroyed
    due to a node deletion, one of its nodes must have been deleted.
    Hence, any node $u$ which lists a triangle $\{u, v, w\}$
    and sees the deletion of one of its two neighbors
    from its adjacency list will stop listing it.
\end{proof}

\begin{observation}[\cite{BC19}]\label{obs:list-clique-triangles}
    If a node lists all the triangles containing it within a clique
    of size $k \geq 3$, then the node lists the clique.
\end{observation}

\begin{proof}
This observation was made in~\cite{BC19}. If
a node lists all triangles containing it within a clique, then
it lists all edges between the nodes in each of these triangles.
From this information, it can compute and list every
clique that may be formed from the set of nodes in these
triangles.
\end{proof}

\begin{theorem}[Theorem 3.3.3~\cite{BC19}]\label{thm:list-node-insert}
    The deterministic $1$-round bandwidth complexity of $\listc(K_s)$
    under node insertions is $O(1)$.
\end{theorem}

\begin{proof}
    By~\cref{lem:node-insertion-triangle}, under node insertions,
    any node $v$ inserted in round $r$ lists every triangle
    containing it and a node $u$ inserted in round $r' > r$
    in $1$-round using $O(1)$-bandwidth, deterministically. Then,
    suppose $K_s$ is formed in round $j$ after the $j$-th node insertion.
    Each node of $K_s$ in the initial graph lists all subsequent
    triangles formed in $K_s$ from node insertions (or the first
    inserted node can do this if no nodes in $K_s$ were present
    in the initial graph). By~\cref{obs:list-clique-triangles},
    this node lists $K_s$ in round $j$ using $O(1)$ bandwidth.
\end{proof}

\begin{corollary}\label{cor:list}
    Given node insertions, the $i$-th node $v$ inserted for clique $K_s$ lists the
    smaller $K_{s - i + 1}$ clique of $K_s$ composed
    of nodes inserted after $v$.
\end{corollary}

\begin{proof}
    By~\cref{lem:node-insertion-triangle}, each node $v$ lists
    all triangles containing $v$ and incident to nodes inserted after it. Hence, each
    node lists all incident cliques with nodes inserted after it.
\end{proof}

\begin{observation}\label{lem:min-can-list}
    If a node $v$ lists a clique of size $k \geq 3$, then
    $v$ lists all triangles containing $v$ in the clique.
\end{observation}

\begin{proof}
This observation can be easily shown:
if a node lists a clique, then it lists all nodes
contained in the clique by definition. Then,
because the node knows that such a clique exists, it can
find all combinations of $3$ nodes in the clique and list
such combinations as all triangles in the clique.
\end{proof}

\begin{observation}\label{lem:triangle-deletion}
    Any node/edge deletion in a clique results in (at least)
    one adjacent triangle deletion (where the triangle is 
    composed of nodes in the clique)
    for every node in the clique.
\end{observation}

\begin{proof}
    A clique consists of a set of triangles which contains all subsets
    of $3$ vertices from the clique. Suppose the edge deletion
    occurs between vertices $u$ and $v$, then any other vertex
    $w \neq u, v \in C$, where $C$ is the clique, must be adjacent to $u$ and $v$,
    and hence formed
    a triangle with $u$ and $v$. Since the edge between $u$ and $v$
    is destroyed, $w$ cannot form a triangle with $u$ and $v$, and
    hence, a triangle adjacent to $w$ is destroyed.
\end{proof}

Given the above observations and the relevant lemmas and theorems from~\cite{BC19},
we are now ready to prove the main theorem
of this section which resolves Open Question $4$ of~\cite{BC19}.

\insertsthm*

\begin{proof}
    To begin, we stipulate that no node lists a clique if it did not list
    the clique previously after receiving a deletion message.
    \cref{lem:node-insertion-triangle}
    shows that every node $v$ which is inserted in round $r$ lists every triangle
    containing it and a node $u$ inserted in round $r' > r$. Using this
    observation, we handle the updates in the following
    way (reiterating the strategy employed by our key observations).
    First, we show that any created clique is listed by at least one of 
    the nodes it contains.
    By~\cref{lem:node-insertion-triangle} and~\cref{obs:list-clique-triangles}, 
    the node that is inserted first in a $K_s$ clique
    lists the $K_s$ clique that is created by node insertions.

    We now show the crux of our proof
    that each node which lists $K_s$ also successfully
    stops listing it after it is destroyed.
    As proven previously in~\cref{lem:node-insertion-triangle}, 
    any node which lists a triangle stops listing it when it is destroyed. 
    Given any edge deletion or node deletion
    which destroys a clique $C$, it must destroy at least one triangle
    adjacent to every node in the clique $C$ by~\cref{lem:triangle-deletion}.
    Every node which lists the clique lists all triangles containing it
    in the clique by~\cref{lem:min-can-list}.
    Then, every node $v$ which lists a clique will know if one of
    its listed triangles is destroyed; then $v$
    will also know of the clique that is destroyed if a triangle
    that makes up the clique is destroyed.

    Finally, to conclude our proof, we explain one additional scenario
    that may occur: a clique which is destroyed may be added back again later
    on. Any clique which is destroyed due to an edge deletion
    cannot be added again unless a node deletion followed
    by another node insertion occurs since
    we are only allowed edge insertions associated with node
    insertions. Thus, by~\cref{cor:list}, the earliest
    node(s) (if such nodes existed in the initial graph)
    in this newly formed clique lists this clique.
    
    Hence, we have proven that any created clique is listed
    by at least one of its nodes and
    that any node which lists the clique also stops listing it after it is destroyed,
    concluding our proof of our theorem. The 
    bandwidth is $O(1)$ since each 
    node sends either no message, \new, 
    or \del to each of its neighbors each round.
\end{proof}

\section{Wedge Listing}\label{sec:wedges}

An \emph{induced wedge} $(u, v, w)$ in the input graph $G = (V, E)$ 
is a path of length $2$ where edges $(u, v)$ and $(v, w)$ exist 
and no edge exists between $u$ and $w$ 
in the induced subgraph consisting of nodes $\{u, v, w\} \subseteq V$. 
We denote a wedge by $\Gamma$.
In this section, we provide an algorithm
listing \emph{induced} wedges. 
Listing non-induced wedges is a much simpler
problem since each node can simply list pairs of its
adjacent neighbors without knowing whether an edge exists between each pair. 
For simplicity, from here onward, we say ``wedge" to mean induced wedge.
\cref{lem:wedge-listing} is the main theorem in this section for
listing wedges. We give the pseudocode for the algorithm
used in the proof of this theorem in~\cref{alg:list-wedges}.
Our algorithm uses listing triangles as a subroutine since
an induced wedge can be formed from an edge deletion to a triangle.
Listing wedges may be useful as a subroutine in
future work that proves additional bounds for listing
other small subgraphs using small bandwidth. 

\SetKwFunction{FnListWedges}{ListWedges}

\begin{algorithm}[!t]
    \caption{\label{alg:list-wedges} Listing Wedges}
    \textbf{Input:} Update $U$ which can be a node deletion, edge insertion, or edge deletion.\\ 
    \textbf{Output:} Each induced wedge is listed by at least one of its nodes and no set of 
    three nodes which is not a wedge is wrongly listed as a wedge.\\
    \Fn{\FnListWedges{$U$}}{
        \If{$U$ adds $u$ to $N(v)$}{
            $v$ sends $(\new, ID_u)$ to all $w \in N(v)$.\label{wedge:edge-ins}\\
            \For{every wedge $(v, x, u)$ that $v$ lists}{
                $v$ stops listing wedge $(v, x, u)$. \label{wedge:stop-list-adj}\\
                $v$ starts listing triangle $\{v, x, u\}$. \label{wedge:ins-triangle-list}
            }
        } \ElseIf{$U$ deletes $u$ from $N(v)$}{
            $v$ sends $(\del, ID_u)$ to all $w \in N(v)$.\label{wedge:edge-del}\\
            \For{every triangle $\{v, x, u\}$ that $v$ lists}{
                $v$ starts listing new wedge $(v, x, u)$.\label{wedge:del-list-new-wedge}\\
                $v$ stops listing triangle $\{v, x, u\}$.\label{wedge:stop-del-list-triangle}
            }

            \For{every wedge $(x, v, u)$ and/or
            $(u, v, x)$ that $v$ lists}{\label{wedge:del-destroyed-wedge}
                $v$ stops listing $(u, v, x)$ and/or 
                $(x, v, u)$.\label{wedge:del-stop-list}\\
            }
        } 
        \If{node $x$ receives $(\new, ID_u)$ from $v$}{\label{wedge:receive-ins}
            \If{$x$ is not adjacent to $u$}{\label{wedge:not-adj}
                $x$ starts listing new wedge $(x, v, u)$.\label{wedge:list-new}\\
            }
            \If{$x$ lists wedge $(v, x, u)$}{\label{wedge:lists-wedge-ins}
                $x$ stops listing wedge $(v, x, u)$.\label{wedge:stops-list-wedge-ins}\\
                $x$ starts listing triangle $\{v, x, u\}$.\label{wedge:starts-list-triangle-ins}\\
            }%
        }
        \If{$x$ receives $(\del, ID_u)$ from $v$}{\label{wedge:receive-del}
            \If{$x$ lists wedge $(x, v, u)$}{\label{wedge:listed-wedge}
                $x$ stops listing wedge $(x, v, u)$.\label{wedge:stop-list}\\
            }
            \If{$x$ lists triangle $\{x, v, u\}$}{\label{wedge:list-triang}
                $x$ starts listing new wedge $(v, x, u)$.\label{wedge:del-new}\\
                $x$ stops listing triangle $\{v, x, u\}$.\label{wedge:del-new-stop-triangle}
            }
        }

    }
\end{algorithm}

We describe our algorithm below. All messages are sent \emph{in the same round}; multiple
messages sent between $u$ and $v$ are concatenated with each other and sent as one message.
On an update $U$, in~\cref{alg:list-wedges}, each node sends $O(1)$ bits, 
indicating whether a neighbor
is deleted or added to their neighbors.
Let us denote these bits by $\del$ and $\new$, respectively.
Furthermore, they send the ID of the neighbor that is deleted or inserted.
Let this ID be $\id_u$ for node $u$.
This procedure is shown in~\cref{wedge:edge-ins,wedge:edge-del} 
in~\cref{alg:list-wedges} and all of the message sending is done in $1$ round.
An edge insertion can form a triangle; thus, in the same round as the $(\new, ID_u)$ message,
node $v$ also starts listing the new triangle 
$\{v, x, u\}$ (\cref{wedge:ins-triangle-list,wedge:stop-list-adj}) 
and stops listing the wedge.
The node/edge deletion can destroy a wedge 
(\cref{wedge:del-destroyed-wedge});
in which case, the node stops listing the 
wedge (\cref{wedge:del-stop-list}).
Similarly, for an edge/node deletion that destroys a triangle, $v$ starts listing wedge %
$(v, x, u)$ for every triangle $\{v, x, u\}$ that $v$ lists (\cref{wedge:del-list-new-wedge,wedge:stop-del-list-triangle}).

Any node $x$ that receives $(\new, ID_u)$ from a neighbor $v$ (\cref{wedge:receive-ins})
lists $(x, v, u)$ as a new wedge (\cref{wedge:list-new}) if it is not adjacent to $u$ (\cref{wedge:not-adj}).
If $x$ lists wedge $(v, x, u)$ (\cref{wedge:lists-wedge-ins}), then $x$ stops listing the wedge (\cref{wedge:stops-list-wedge-ins})
and starts listing triangle $\{v, x, u\}$ (\cref{wedge:starts-list-triangle-ins}).

Any node $x$ that receives $(\del, ID_u)$ from a neighbor $v$ (\cref{wedge:receive-del})
first checks if a wedge it lists is destroyed. If $x$ lists wedge $(x, v, u)$, then 
it stops listing wedge $(x, v, u)$ since it is destroyed (\cref{wedge:stop-list}).
Then, it checks if the deletion destroys a triangle it lists. If it lists triangle $\{x, v, u\}$,
then it lists a new wedge $(v, x, u)$ since a destroyed triangle creates a new wedge (\cref{wedge:del-new,wedge:del-new-stop-triangle})
and stops listing the triangle.

\wedgethm*

\begin{proof}
    We prove via induction that this algorithm achieves the following guarantees
    \begin{enumerate}[(a)]
        \item any wedge is listed by at least one node
        \item and no sets of 
    three nodes that is not a wedge is incorrectly listed as a wedge
    \end{enumerate} 
    after any edge insertion
    or node/edge deletion.
    In the base case, the graph is either the input graph and every
    wedge in the graph is listed by at least one node or the graph is empty.
    We assume that at step $j$, each wedge in the graph is
    listed by some node in the wedge. We show that all wedges formed
    in round $j+1$ by some update is listed by at least one node in the wedge; 
    then, we show that
    each wedge that is destroyed in round $j + 1$ is no longer listed by any node.
    We prove this by casework over the type of update:

    \begin{itemize}
        \item Edge insertion: Given an edge insertion $\{u, v\}$, node $v$ sees some
            node $u$ added to its adjacency list and sends $(\new, \id_u)$ to all
            its neighbors in $N(v)$ (similarly for $u$). Suppose
            $x \in N(v)$ is a neighbor of $v$ but not of $u$. Then, $x$ can distinguish
            between triangle $\{x, v, u\}$ and wedge $(x, v, u)$ by seeing whether $u$
            is in its adjacency list (i.e. whether edge $\{x, u\}$ exists). Hence,
            $x$ lists wedge $(x, v, u)$. %

            An edge insertion can cause a wedge to be destroyed if it forms a triangle from a wedge. 
            Suppose $(x, v, u)$ is a wedge that becomes a triangle due to edge insertion $\{x, u\}$.
            If $x$ lists the wedge, then $x$ knows $\{x, v, u\}$ is now a triangle because it see
            $u$ added to its adjacency list. The same argument holds for $u$. If $v$ lists the 
            wedge, then $v$ receives $(\new, ID_x)$ and $(\new, ID_u)$ and knows that $\{x, v, u\}$
            is now a triangle and stops listing it as a wedge.
        \item Edge deletion: Suppose an edge deletion $\{u, v\}$ destroys wedge $(x, v, u)$.
            Since $u$ and $v$ can see that the other is no longer a neighbor, if
            they could list wedge $(x, v, u)$, they now know that $(x, v, u)$ has been destroyed and stops
            listing it. Node $v$ sends to $x$ the tuple $(\del, \id_u)$ using $O(\log n)$ bandwidth and
            $1$ round of communication. Hence, after this one round of communication,
            $x$, if it can list $(x, v, u)$, now knows $(x, v, u)$ has been destroyed and stops listing it.

            We now show that if an edge deletion destroys a triangle, then at least one of its nodes
            lists the new wedge.
            Under our current set of update operations, any triangle
            is formed after $3$ edge insertions. After the first two
            edge insertions, at least one node $x$ lists the wedge that is formed
            by our argument provided above for edge insertions. Then, after the third insertion,
            $x$ lists the triangle formed (and stops listing the wedge).
            Suppose without loss of generality that node $x$ lists triangle $\{x, v, u\}$.
            If $\{x, u\}$ is deleted, $x$ sees $u$ removed from its adjacency list
            and lists $(x, v, u)$ as a new wedge. The same argument holds
            for edge deletion $(x, v)$. If instead $(u, v)$ is deleted, 
            then node $x$ receives $(\del, ID_v)$ and $(\del, ID_u)$
            and lists $(v, x, u)$ as a new wedge. %

        \item Node deletion: Suppose wlog $u$ is deleted and there exists some
            wedge containing $u$, $v$, and $w$. Then, if $u$ is deleted and
            the wedge consists of edges
            $\{v, u\}$ and $\{u, w\}$, then both $v$ and $w$ notice that $u$
            has been removed from their adjacency lists. Hence, if either of the nodes
            listed the wedge, they would know it is destroyed and stops listing it.
            If the wedge consists
            of edges $\{u, v\}$ (resp. $\{u, w\}$) and $\{v, w\}$, then $w$ would know
        of the deletion of edge $\{u, v\}$ (resp. $v$ of edge $\{u, w\}$) after
        receiving messages since $v$ (resp. $w$) sent to all its neighbors $(\del, \id_u)$.
    \end{itemize}

    Hence, in step $j+1$ after the $j+1$-st synchronous round, all wedges formed in step $j+1$
    is listed by at least one vertex. The
    bandwidth is $O(\log n)$ since the 
    IDs of vertices are $O(\log n)$ sized.
\end{proof}

\section{Batched Triangle Listing}\label{sec:batched-triangle}

In this section, we extend our algorithms in the previous sections to handle
batches of more than one update. Our results are in the batched model (formally defined in~\cref{sec:prelims}) where in
each batch of updates, $O(1)$ updates occur \emph{adjacent to any node in the graph}.
Such a model is realistic since often many
updates can occur in total over an entire
real-world network (such a social network) but the
updates adjacent to each node in the network are often few in number.

We show that using our techniques above, we can also perform $\listc(K_3)$ in
$1$-round using $O(\log n)$ bandwidth. When provided a batch of updates,
each node with one or more new neighbors sends the ID of the new neighbor(s) to all its
neighbors. 
Furthermore, each node sends the IDs of all nodes deleted from its adjacency list to 
all its neighbors. We show that this simple algorithm allows us to perform $\listc(K_3)$
in one round under any type of update. The pseudocode for this 
algorithm is given in~\cref{alg:batched-triangle}. All messages are sent in the same round
in~\cref{alg:batched-triangle} and multiple messages sent (in the pseudocode)
between two nodes are concatenated into the same message.

\SetKwFunction{FnBListTriangle}{BatchedListTriangles}
\begin{algorithm}[!t]
    \caption{\label{alg:batched-triangle} Batched Listing Triangles}
    \textbf{Input:} Updates $\mathcal{U}$ which is a set of updates consisting of
    node insertions, node deletions, edge insertions, and/or edge deletions.\\ 
    \textbf{Output:} Each triangle is listed by at least one of its nodes and no set of 
    $3$ nodes which is not a triangle is wrongly listed.\\
    \Fn{\FnBListTriangle{$\mathcal{U}$}}{
        \If{$\mathcal{U}$ adds node(s) $u \in I_v$ to $N(v)$}{\label{batched-triangle:add-adj}
            $v$ sends $(\new, \{ID_u \mid u \in I_v\})$ to all $w \in N(v)$.\label{batched-triangle:send-new}\\
        } 
        \If{$w$ receives $(\new, S_v)$ from $v$}{\label{batched-triangle:receive-new}
            \For{$ID_u \in S_v$}{
                \If{$u \in N(w)$}{
                $w$ starts listing new triangle $\{u, v, w\}$.\label{batched-triangle:start-listing}\\
                }
            }
        }
        \If{$\mathcal{U}$ deletes node(s) $u \in D_v$ from $N(v)$}{\label{batched-triangle:delete-adj}
            $v$ sends $(\del, \{ID_u \mid u \in D_v\})$ to all $w \in N(v)$.\label{batched-triangle:send-del}\\
            \For{$u \in D_v$}{
                \For{$w \in N(v) \cup D_v$}{
                    \If{$v$ lists triangle $\{u, v, w\}$}{\label{batched-triangle:list-triangle}
                        $v$ stops listing triangle $\{u, v, w\}$.\label{batched-triangle:stop-listing}
                    }
                }
            }
        }
        \If{$w$ receives $(\del, S_v)$ from $v$}{\label{batched-triangle:receive-2-del}
            \For{$ID_u \in S_v$}{
                \For{all triangles $\{u, v, w\}$ that $w$ lists}{\label{batched-triangle:list-new-2}
                    $w$ stops listing triangle $\{u, v, w\}$.\label{batched-triangle:stop-list-2}
                }
            }
        }
    }
\end{algorithm}

\triangleslist*

\begin{proof}
    For any inserted edge $\{u, v\}$, node $u$ sends $ID_v$ to 
    all its neighbors (and similarly for node $v$). Any neighbor $w$ of 
    both $u$ and $v$ (and where neither $u$ or $v$ are deleted from its adjacency list)
    lists the new triangle $\{u, v, w\}$. Edges $\{u, w\}$ and/or $\{v, w\}$ may also be 
    new edges and the proof still holds.    
    
    Now, we consider node insertions. An inserted node, $u$, creates a
    triangle if it is adjacent to two newly inserted edges, $\{u, v\}, \{u, w\}$, and the edge
    $\{v, w\}$ exists in the graph or is newly inserted. If $\{v, w\}$ already exists, then
    $v$ sends $\id_u$ to $w$ and $w$ sends $\id_u$ to $v$; thus, both $v$ and $w$ now list
    triangle $\{u, v, w\}$. If $\{v, w\}$ is newly inserted, then $v$ and $w$ would still send 
    $\id_u$ to each other and also send each other's IDs to $u$. Thus, all three nodes now list the triangle. 
    
    Finally, we consider node/edge deletions. For any destroyed triangle $\{u, v, w\}$, suppose
    wlog that $u$ lists it. Either $u$ is adjacent to an edge deletion, in which case, $u$ 
    stops listing $\{u, v, w\}$; or, $u$ is still adjacent to both $v$ and $w$ and edge $\{v, w\}$
    is deleted. In the second case, $v$ and $w$ both send $u$ each other's ID and node $u$ stops 
    listing $\{u, v, w\}$. A node deletion is incident to all remaining nodes of the triangle and 
    so the remaining nodes stop listing the triangle.

    Since each node is adjacent to $O(1)$ updates and sends a message of size equal to the number of
    adjacent updates times the size of each node's ID (where the size of each ID is $O(\log n)$), 
    the total bandwidth used is $O(\log n)$.
\end{proof}

\section{Batched Clique Listing}\label{sec:batched-clique}

A slightly modified algorithm allows us to perform $\listc(K_s)$ under
the same constraints; however, the proof is more
involved than the $K_3$ case
but uses the same concepts we developed in the previous sections.
The modification computes the cliques using the triangles listed after a batch of 
updates similar to~\cref{alg:list-cliques}.

Unlike the case for triangles, here, as in the case for counting cliques under single updates,
\cref{thm:list-clique}, we cannot handle edge insertions in $O(\log n)$ bandwidth. We give the 
new pseudocode for this procedure in~\cref{alg:batched-clique}. This part of the procedure
is identical to that given in~\cref{alg:list-cliques}.

\SetKwFunction{FnBListClique}{BatchedListCliques}
\begin{algorithm}[!t]
    \caption{\label{alg:batched-clique} Batched Listing Cliques}
    \textbf{Input:} Updates $\mathcal{U}$ which is a set of updates consisting of
    node insertions, node deletions, and/or edge deletions.\\ 
    \textbf{Output:} Each $K_s$ is listed by at least one of its nodes and no set of 
    $s$ nodes which is not a $K_s$ is wrongly listed.\\
    \Fn{\FnBListClique{$\mathcal{U}$}}{
        Run \FnBListTriangle{$\mathcal{U}$} (\cref{alg:batched-triangle}).\\
        \For{all $v \in V$}{\label{batched-clique:iterate-node}
            \For{every $S \subseteq N(v)$ of $s-1$ neighbors}{\label{batched-clique:every-subset}
                \If{$v$ lists $\{v\} \cup A$ as a triangle for every subset of $2$ nodes $A = \{a, b\}\subseteq S$}{\label{batched-clique:all-subsets-triangle}
                    $v$ lists $S \cup \{v\}$ as a $K_s$.\label{batched-clique:list-new-clique}
                }
            }
            $v$ stops listing every clique containing a destroyed triangle.\\
        }
    }
\end{algorithm}

To prove the main theorem in this section,~\cref{lem:batched-clique}, we only need to prove a slightly
different version of \cref{obs:list-clique-triangles} that still holds in this setting. Namely, 
we modify~\cref{obs:list-clique-triangles} so that it holds for triangles created by two or more
updates in the same round.

\begin{observation}\label{lem:batched-obs-1}
    Under node insertions and node/edge deletions,
    suppose $v$ is inserted in round $r$ and $u$ is inserted in round
    $r' \geq r$. If $\{u, v, w\}$ is a triangle in round $r'$, then in 
    round $r'$, $u$ lists $\{u, v, w\}$ using $O(1)$ bandwidth. 
\end{observation}

\begin{proof}
    The key difference between this observation and~\cref{obs:list-clique-triangles} is that the observation
    holds when $u$ is inserted in the same round $r$ as $v$. Suppose triangle $\{u, v, w\}$ is created in round $r'$.
    We first consider the case when $r' > r$. In this case, $w$ sends $ID_u$ to $v$ and $v$ sees $u$ in its adjacency list.
    Thus, $v$ lists triangle $\{u, v, w\}$. Now, suppose $r' = r$. In this case, $w$ still sends $ID_u$ to $v$ and $v$ sees
    $u$ in its adjacency list; thus, $v$ lists $\{u, v, w\}$. 

    There are multiple scenarios for edge deletions that destroy triangle $\{u, v, w\}$. First, if one edge, $\{u, v\}$,
    is deleted, then $u$ and $v$ send each other's IDs to $w$ in deletion messages. Thus, $w$ stops listing $\{u, v, w\}$ 
    if it previously listed $\{u, v, w\}$. If two edges are deleted, then, every node is adjacent to at least one deletion
    and will stop listing the triangle. A node deletion causes at least one adjacent edge to be deleted from every remaining
    node in the triangle; thus, any remaining node(s) will stop listing the triangle.
\end{proof}

\begin{observation}\label{lem:batched-obs-2}
    Under node insertions and node/edge deletions,
    if triangle $\{u, v, w\}$ is created in round $r$ from two or more node insertions, then all 
    three nodes in the triangle lists it.
\end{observation}

\begin{proof}
    Without loss of generality, suppose $u$ and $v$ are inserted in round $r$. Then, $u$ and $v$ sees each other 
    in their adjacency lists. Node $w$ sends $ID_u$ and $ID_v$ to both $u$ and $v$. Hence, both $u$ and $v$ lists 
    triangle $\{u, v, w\}$. Node $w$ receives $ID_u$ from $v$ and $ID_v$ from $u$ and sees both $u$ and $v$ in its
    adjacency list. Hence, node $w$ lists $\{u, v, w\}$. When all three nodes are inserted in the same round, 
    each node receive the other nodes' IDs via messages and successfully lists the triangle.
    The remaining parts of this proof is identical to the second part of the proof of~\cref{lem:batched-obs-1}.
\end{proof}

\cliqueslist*

\begin{proof}
    The key to this proof, as in the case for one update, is that there exists one node that
    lists all triangles incident to it
    in a new clique and hence lists the clique. 
    By~\cref{lem:batched-obs-1}, the node inserted in the earliest round for each $K_s$
    lists it. If all nodes in a $K_s$ are inserted in the current round, then 
    by~\cref{lem:batched-obs-2} and~\cref{obs:list-clique-triangles}, all nodes in the 
    clique lists it.

    We now show that each node which lists a clique will stop listing it when it is destroyed.
    By~\cref{lem:triangle-deletion}, at least one triangle is destroyed for every node in a clique $C$
    if an edge deletion destroys $C$. Then, a triangle is destroyed by either an
    edge deletion or a node deletion (or both). For a node deletion $v$, an edge incident to every other
    node in any triangle containing $v$ will be deleted. Thus, each node listing the triangle containing
    $v$ will stop listing it when $v$ is deleted. For any edge deletion $\{u, v\}$ that destroys
    a triangle $\{u, v, w\}$, if edges $\{u, w\}$ and $\{v, w\}$ are not deleted, then $w$ receives
    a delete message from $u$ and $v$ listing each other's ID. Then, $w$ knows that 
    triangle $\{u, v, w\}$ is destroyed and stops listing it. 
    We know every node $v$ that lists $C$ lists all of the 
    triangles composed of nodes in $C$ that are
    incident to $v$. We just showed that every node that lists a triangle successfully stops listing 
    it when it is destroyed. Then, it follows 
    that each node which lists $C$ stops listing it when it is destroyed.
    
    Since we are guaranteed $O(1)$ incident updates to every node, by~\cref{alg:batched-clique}, 
    each node sends $O(1)$ IDs for each batch of updates and the size of each message
    is still $O(\log n)$.
\end{proof}

\section{Open Questions}

We hope the observations we make in this paper will be useful for future work.
The open question we find the most interesting regarding work in this area
are techniques to list other induced subgraphs that are not cliques
in one round under very small bandwidth. A preliminary
look at listing $k$-paths and $k$-cycles 
in $1$-round under $O(1)$-bandwidth
proves to be impossible if node deletions are allowed due to the following
reason. Suppose we are given a subgraph $H$ with radius greater than $2$ and
suppose $v$ is currently listing $H$. Then a node deletion of a node $u$ that destroys the subgraph,
where $u$ is at a distance greater than $2$ from $v$, 
will not be detected by $v$ in one round.
However, perhaps results could found for other types of subgraphs and if there are no
node/edge deletions.
Another interesting open question is whether these results could be
extended to \emph{arbitrarily} large number of edge/node updates, specifically extending
our results given in the batched setting.

\section*{Acknowledgements} We are very grateful to our anonymous reviewers whose suggestions
greatly improved the presentation of our results. In particular, we thank one anonymous reviewer for the 
simple argument that $k$-path and $k$-cycle listing in one round is difficult for radius greater than $2$.

\bibliographystyle{alpha}
\bibliography{ref}
\end{document}